\documentclass[onecolumn,journal]{IEEEtran}
\ifCLASSINFOpdf
\else
\fi

\usepackage[T1]{fontenc}
\usepackage[latin9]{inputenc}
\usepackage{color}
\usepackage{array}
\usepackage{float}
\usepackage{mathrsfs}
\usepackage{mathtools}
\usepackage{amsthm}
\usepackage{amstext}
\usepackage{amssymb}
\usepackage{stmaryrd}
\usepackage{graphicx}
\usepackage{pgfplots}

\makeatletter

\theoremstyle{plain}
\newtheorem{thm}{\protect\theoremname}
\theoremstyle{plain}
\newtheorem{prop}[thm]{\protect\propositionname}
\theoremstyle{plain}

\theoremstyle{plain}
\newtheorem{lem}[thm]{\protect\lemmaname}
\theoremstyle{definition}
\newtheorem{example}[thm]{\protect\examplename}
\theoremstyle{definition}

\theoremstyle{plain}
\newtheorem{remark}[thm]{\protect\remarkname}

\AtBeginDocument{
	
}

\makeatother

  \providecommand{\corollaryname}{Corollary}
  \providecommand{\examplename}{Example}
  \providecommand{\lemmaname}{Lemma}
  \providecommand{\propositionname}{Proposition}
  \providecommand{\theoremname}{Theorem}
  \providecommand{\definitionname}{Definition}
  \providecommand{\remarkname}{Remark}

\newcommand{\ceil}[1]{{\left\lceil {#1} \right\rceil}}
\newcommand{\floor}[1]{{\left\lfloor {#1} \right\rfloor}}

\newcommand{\Bk}[1]{{\Big( {#1} \Big) }}
\begin{document}

%
\title{Weierstrass Semigroups from Kummer Extensions}
%
%
%

\author{Shudi~Yang and~Chuangqiang~Hu   
\thanks{S. Yang is with the School of Mathematical
	Sciences, Qufu Normal University, Shandong 273165, P.R.China. \protect\\
		\quad C. Hu is with the School of Mathematics, Sun Yat-sen University, Guangzhou 510275, P.R.China.\protect\\	
	\protect\\
	E-mail: huchq@mail2.sysu.edu.cn,~{yangshd3@mail2.sysu.edu.cn}} \protect\\
\thanks{Manuscript received *********; revised ********.}
}

\maketitle

\begin{abstract}
The Weierstrass semigroups and pure gaps can be helpful in constructing codes with better parameters. In this paper, we investigate explicitly the minimal generating set of the Weierstrass semigroups associated with several totally ramified places over arbitrary Kummer extensions. Applying the techniques provided by Matthews in her previous work, we extend the results of specific Kummer extensions studied in the literature. Some examples are included to illustrate our results.

\end{abstract}

\begin{IEEEkeywords}
Kummer extension, Weierstrass semigroup, Weierstrass gap
\end{IEEEkeywords}

%
\IEEEpeerreviewmaketitle

\section{Introduction}
%
%
%
%

\IEEEPARstart{S}{ince} Goppa~\cite{goppa1977codes} constructed algebraic geometric (AG) codes from several rational places, the study of AG codes becomes an important instrument in coding theory. For a given AG code, the famous Riemann-Roch theorem gives a non-trivial lower bound, named Goppa bound, for the minimum distance in a very general setting~\cite{stichtenoth2009algebraic}.
Garcia, Kim and Lax improved the Goppa bound using arithmetical structure of the Weierstrass gaps at one place in~\cite{garcia1993consecutive,garcia1992goppa}.
Homma and Kim~\cite{Homma2001Goppa} introduced the concept of pure gaps and demonstrated a similar result for a divisor concerning a pair of places. And this was generalized to several places by Carvalho and Torres in~\cite{carvalho2005goppa}.

Weierstrass semigroups and pure gaps are of significant uses in the construction and analysis of AG codes. They would be applied to obtain codes with better parameters (see~\cite{HuYang2016Kummer,Masuda2}). We see that the Weierstrass semigroup associated with several distinct places is a generalization of the classical studied one at only one place. More information can be found in~\cite{niederreiter2001rational}.

Weierstrass semigroups over specific Kummer extensions were well-studied in the literature. For instance,
Matthews~\cite{matthews2004weierstrass} investigated the Weierstrass semigroup of any $ l $-tuple of collinear places on a Hermitian curve. In~\cite{matthews2005weierstrass}, Matthews generalized the results of~\cite{carvalho2005goppa,matthews2004weierstrass} by determining the Weierstrass semigroup of
any $ l $-tuple rational places on the quotient of the Hermitian curve defined by the equation $ y^q + y = x^m $ over
$ \mathbb{F}_{q^2} $ where $  m>2 $ is a divisor of $ q + 1 $.
However, little is known for general Kummer extensions, except that,  Masuda, Quoos and Sep{\'u}lveda~\cite{Masuda2} recently described the Weierstrass semigroups
of one- and two-places.

In this paper, our main interest will be in the research of the Weierstrass semigroup of
any $ l $-tuple rational places over arbitrary Kummer extensions.
This work is strongly inspired by the study of
\cite{Masuda2,matthews2004weierstrass,matthews2005weierstrass}. We shall explicitly calculate the minimal generating set of the Weierstrass semigroups by employing the techniques provided by Matthews in \cite{matthews2004weierstrass,matthews2005weierstrass}. At this point, we extend the results of~\cite{Masuda2,matthews2004weierstrass,matthews2005weierstrass}. We mention that our results can be employed to get linear codes attaining good or better records on the parameters, see~\cite{HuYang2016Kummer,Masuda2} for more details.

The remaider of the paper is organized as follows. In Section~\ref{sec:prelimi} we briefly recall some notations and preliminary results over arbitrary function fields. Section~\ref{sec:main results} focuses on the determination of the minimal generating sets of the Weierstrass semigroups over Kummer extensions. Finally, in Section~\ref{sec:Examples} we exhibit some examples by using our main results.

\section{Preliminaries}\label{sec:prelimi}

In this section, we introduce notations and present some basic facts on the minimal generating set of Weierstrass semigroups of distinct rational places over arbitrary function fields.

Let $ q  $ be a power of a prime $p$ and $ \mathbb{F}_{q} $ be a finite field with $ q $ elements.  We denote by $ F $ a function field with genus $ g $ over $ \mathbb{F}_q$ and by
$ \mathbb{P}_F $ the set of places of $ F $. The free abelian group generated by the places of $ F $ is denoted by $ \mathcal{D}_F $, whose element is called a divisor. Assume that  $ D=\sum_{P\in {\mathbb{P}_F}} n_P P$ is a divisor such that almost all $ n_P=0 $, then the degree of $ D $ is $ \deg(D)= \sum_{P\in \mathbb{P}_F} n_P $.
For a function $ f \in F $, the divisor of $ f $ will be denoted by $ (f) $ and the divisor of poles of $ f $ will be denoted by $ (f)_{\infty} $.

We introduce some notations concerning the Weierstrass semigroups. Given $ l $ distinct rational places of $ F $, named $ Q_1,\cdots, Q_l $, the Weierstrass semigroup
$ H(Q_1,\cdots, Q_l) $ is defined by
\[
\Big\{(s_1,\cdots, s_l)\in \mathbb{N}_0^l~\Big|~\exists f\in F~ \text{with}~ (f)_{\infty}=\sum_{i=1}^l s_i Q_i  \Big\},
\]
and the Weierstrass gap set $  G(Q_1,\cdots, Q_l)  $ is defined by $ \mathbb{N}_0^l \backslash H(Q_1,\cdots, Q_l) $, where $ \mathbb{N}_0 := \mathbb{N}\cup \{0\} $ denotes the set of nonnegative integers.
When comparing elements of $ \mathbb{N}_0^l $, we always employ a partial order $ \preceq $ defined by $ (n_1,\cdots,n_l) \preceq (p_1,\cdots,p_l) $ if and only if $ n_i\leqslant p_i $ for all $ i $, $ 1 \leqslant i \leqslant l $.

In order to present the minimal generating set for Weierstrass semigroups, more symbols should be described. To begin with, set $ \Gamma(Q_1):= H(Q_1) $.
For $ l\geqslant 2 $, define
\begin{equation*}
\Gamma(Q_1,\cdots,Q_l):= \Big\{ \mathbf{n} \in \mathbb{N}^l ~\Big|~ \mathbf{n} \text{ is minimal in } \{ \mathbf{p}\in H(Q_1,\cdots,Q_l)~|~ p_i=n_i\} \text{ for some } i, 1\leqslant i \leqslant l \Big \}.
\end{equation*}
For $ l\geqslant 1 $, define
\begin{equation*}
\tilde{\Gamma}(Q_1,\cdots,Q_l):= \Big\{ \mathbf{n}\in \mathbb{N}_0^l ~\Big|~ (n_{i_1},\cdots,n_{i_k}) \in \Gamma(Q_{i_1},\cdots,Q_{i_k}) \text{ for some } k, 1\leqslant k \leqslant l \text{ and } 1\leqslant i_1<\cdots < i_k \leqslant l \Big \}.
\end{equation*}
Here the elements in $ \mathbb{N}^l $ (also in $ \mathbb{N}_0^l $ ) are compared with respect to $ \preceq $. We remark that, for the general case, Matthews~\cite{matthews2004weierstrass,matthews2005weierstrass} proposed the notation of $ \Gamma(Q_1,\cdots,Q_l) $ and $ \tilde{\Gamma}(Q_1,\cdots,Q_l) $, while Kim~\cite{kim} settled the case of $ l=2 $. Here we collect two results from~\cite{matthews2004weierstrass} that will be used in the next section.

\begin{lem}[\cite{matthews2004weierstrass}]\label{lem:gamma_l}
	Let $ \mathbf{n} \in \mathbb{N}^l $. Then $ \mathbf{n} $ is minimal in
	$ \{\mathbf{p}\in H(Q_1,\cdots,Q_l)~|~p_i=n_i \}  $ with respect to $ \preceq $
	for some $ i $, $ 1\leqslant i \leqslant l $, if and only if $ \mathbf{n} $ is minimal in the set
	$ \{\mathbf{p}\in H(Q_1,\cdots,Q_l)~|~p_i=n_i \}  $ with respect to $ \preceq $
	for all $ i $, $ 1\leqslant i \leqslant l $.
\end{lem}

\begin{lem}[\cite{matthews2004weierstrass}]\label{lem:gamma_linGapsets}
	Let $ Q_1,\cdots,Q_l $ be distinct rational places with $ l \geqslant 2 $. Then
	\begin{equation*}
	\Gamma(Q_1,\cdots,Q_l) \subseteq G(Q_1) \times \cdots \times G(Q_l).
	\end{equation*}
\end{lem}

The following theorem, due to Matthews~\cite{matthews2004weierstrass}, generalized the result of Kim~\cite{kim} and showed that the entire Weierstrass semigroup $ H(Q_1,\cdots,Q_l) $ is generated by the set $ \tilde{\Gamma}(Q_1,\cdots,Q_l) $. Thus $ \tilde{\Gamma}(Q_1,\cdots,Q_l) $ is called the minimal generating set of the Weierstrass semigroup $ H(Q_1,\cdots,Q_l) $.

\begin{lem}[\cite{matthews2004weierstrass}]\label{lem:Weisemigp}
	Let $ Q_1,\cdots, Q_l $ be distinct rational places with $ l \geqslant 2 $. Then	
	\begin{equation*}
	H(Q_1,\cdots, Q_l) =\Big\{ \mathrm{lub}\{\mathbf{u_1},\cdots,\mathbf{u_l}\}\in \mathbb{N}_0^l~\Big|
	~\mathbf{u_1},\cdots,\mathbf{u_l} \in \tilde{\Gamma}(Q_1,\cdots,Q_l) \Big\},
	\end{equation*}	
	where $ \mathrm{lub}\{\mathbf{u_1},\cdots,\mathbf{u_l}\} $ represents the least upper bound of vectors $ \mathbf{u_i} = (u_{i_1},\cdots,u_{i_l}) $ in $ \mathbb{N}_0^l$ for all $ i $, $ 1\leqslant i \leqslant l $, defined by
	\begin{equation*}
	\mathrm{lub}\{\mathbf{u_1},\cdots,\mathbf{u_l}\}:=(\max \{u_{1_1},\cdots, u_{l_1}\}, \cdots, \max \{ u_{1_l},\cdots,u_{l_l}\}).
	\end{equation*}	
\end{lem}

Actually, in order to determine $ H(Q_1,\cdots,Q_l) $, one only needs to determine $ \Gamma(Q_{i_1},\cdots,Q_{i_k}) $  for all $ 1\leqslant k \leqslant l$ and $  1\leqslant i_1<\cdots < i_k \leqslant l $. This is precisely what we will consider in the next section.

Now, we turn our attention to a characterization of $ \Gamma(Q_1,Q_2) $ in $  (G(Q_1) \times G(Q_2)) \cap H(Q_1, Q_2) $ established by Homma~\cite{homma1996weierstrass}. Denote the gap sequence at $ Q_1 $ by
$ \beta_1<\beta_2<\cdots<\beta_g $ and that at $ Q_2 $ by $ \beta'_1<\beta'_2<\cdots<\beta'_g $, where $ g $ is the genus of the function field. For each gap $ \beta_i $ at $ Q_1 $, the integer
$ n_{\beta_i} :=\min\{\gamma~|~(\beta_i,\gamma)\in H(Q_1, Q_2)\}  $  is a gap at  $ Q_2 $~\cite{kim}. So there exists a permutation $ \sigma $ of the set $ \{1,2,\cdots,g\} $ such that $ n_{\beta_i}=\beta'_{\sigma(i)} $. The graph of the bijective map between $ G(Q_1) $ and $ G(Q_2) $ is the set
$ \Gamma(Q_1,Q_2)=\{(\beta_i,\beta'_{\sigma(i)})\in \mathbb{N}^2~|~i=1,2,\cdots,g\} $. We recall the following lemma.

\begin{lem}[\cite{homma1996weierstrass}]\label{lem:gammaforinf}
	Let $ \Gamma' $ be a subset of $(G(Q_1) \times G(Q_2)) \cap H(Q_1, Q_2) $. If there exists a permutation $ \tau $ of $ \{1,2,\cdots,g\} $ such that
	$ \Gamma'=\{(\beta_i,\beta'_{\tau(i)})\in \mathbb{N}^2~|~i=1,2,\cdots,g\} $, then
	$ \Gamma'=\Gamma(Q_1,Q_2) $.
\end{lem}

\section{Main results}\label{sec:main results}

 Let  $ m \geqslant 2  $ be an integer coprime with $ p $. In this section, we restrict our attention to the Kummer extension $F_{\mathcal{K}}/{\mathbb{F}_q(x)}$ defined by
$
y^m=f(x)^{\lambda}=\prod_{i=1}^{r}(x-\alpha_i)^{\lambda},
$
where $ r>2 $, $ \gcd(m,r \lambda)=1 $ and the $ \alpha_i $'s are pairwise distinct elements in $ \mathbb{F}_q $ for $ 1\leqslant i\leqslant r $. The function field $ F_{\mathcal{K}} $ has
genus $ g= (r-1)(m-1)/2$. Let $ P_1,\cdots,P_r $ be the places of the rational function field $ F_{\mathcal{K}} $ associated to the zeros of $ x-\alpha_1,\cdots,x-\alpha_r $, respectively, and $ P_{\infty} $ be the unique place at infinity. It then follows from~\cite{Stichtenoth} that they are totally ramified in this extension.

The following proposition describes some principle divisors of a Kummer extension.
\begin{prop}\label{prop:divisor}
	Let $F_{\mathcal{K}}/{\mathbb{F}_q(x)}$ be a Kummer extension defined by
	\begin{equation}\label{eq:Kumext}
	y^m=f(x)^{\lambda}=\prod_{i=1}^{r}(x-\alpha_i)^{\lambda},
	\end{equation}
	where $ \alpha_i \in \mathbb{F}_q $ and $ \gcd(m, r\lambda) =1 $. Then we have the following divisors in $F$:
	\begin{enumerate}
		\item[(1)]	$ (x-\alpha_i)=mP_i-mP_{\infty}$,  for $ 1\leqslant i\leqslant r $;
		\item[(2)] $ (y)=\lambda P_1+\cdots +\lambda P_r-r\lambda P_{\infty} $;
		\item[(3)] $ (f(x))=\sum_{i=1}^r mP_i-rmP_{\infty} $;
	    \item[(4)] $ (z)=  P_1+\cdots + P_r-r P_{\infty} $, where $ z:=y^A f(x)^B $, and $ A $, $ B $
	    are integers such that $ A \lambda+Bm =1 $.
	\end{enumerate}	
	
\end{prop}

Denote by $ \lfloor x \rfloor $ the largest integer not greater than $ x  $ and by $ \lceil x \rceil $ the smallest integer not less than $ x $.  Masuda, Quoos and Sep{\'u}lveda~\cite{Masuda2} determined the Weierstrass semigroups related to one and two rational places.

\begin{lem}[\cite{Masuda2}]\label{lem:Weier_gap at one place}
	Let $F/K$ be the Kummer extension given by~\eqref{eq:Kumext}.  Suppose that $ P_{\infty} $ and $ P_1 $ are rational places of $ F $. Then
	\begin{align*}
	H(P_1) &=\mathbb{N}_0\backslash 	\Big\{  mk+j \in \mathbb{N}~\Big|
	~1\leqslant j \leqslant m-1- \left \lfloor \dfrac{m}{r} \right \rfloor,  ~0\leqslant k \leqslant r-2- \left \lfloor \dfrac{rj}{m} \right \rfloor \Big	\},\\
	H(P_{\infty}) &= \langle m,r\rangle.
	\end{align*}
\end{lem}

\begin{lem}[\cite{Masuda2}]\label{lem:Gammap2inf}
		Let $F/K$ be the Kummer extension given by~\eqref{eq:Kumext}. Suppose that $ P_{\infty} $ and $ P_1 $ are rational places of $ F $. Then
	\begin{align*}
	\Gamma(P_{\infty},P_1) =\Big\{ & (mk_0-rj, mk_1+j) \in \mathbb{N}^2~\Big|~1\leqslant j \leqslant m-1-\floor{\dfrac{m}{r}},\\
 & ~k_0+k_1=r-1,~ k_0 \geqslant \ceil{\frac{rj}{m}},  ~k_1\geqslant 0 \Big\}.
   \end{align*}
\end{lem}

The following results are extensions of Lemma~\ref{lem:Gammap2inf}, which will demonstrate the minimal generating set of Weierstrass semigroups of arbitrary rational places. We start with a simple case that provides
a crucial ingredient in the proof of the general case.


\begin{thm}\label{thm:Gamma2points}
	Let $F/K$ be the Kummer extension given by~\eqref{eq:Kumext} with genus $g \geqslant 1 $, $ P_1 $ and $  P_2 $ be two totally ramified places in $ \mathbb{P}_{F} $. Then
	\begin{align*}
	\Gamma(P_1,P_2) =\Big \{&
	(mk_1+j, m k_2 +j)\in \mathbb{N}^2 ~\Big|
	~1\leqslant j \leqslant m-1- \left \lfloor \dfrac{m}{r} \right \rfloor, \\
	& ~k_1,k_2 \geqslant 0,~ k_1+k_2=r-2- \left \lfloor \dfrac{rj}{m} \right \rfloor	 \Big\}.
	\end{align*}
\end{thm}
\begin{proof}
	Let $ \Gamma'=  \Big\{
	(mk_1+j, m k_2 +j)\in \mathbb{N}^2 ~\Big|
	~1\leqslant j \leqslant m-1- \left \lfloor \dfrac{m}{r} \right \rfloor, ~k_1,k_2 \geqslant0, ~ k_1+k_2=r-2- \left \lfloor \dfrac{rj}{m} \right \rfloor	\Big\}$.
	The conditions $ k_1,k_2 \geqslant0 $ and $  k_1+k_2=r-2- \left \lfloor \dfrac{rj}{m} \right \rfloor $ implies that $ 0\leqslant k_i \leqslant r-2- \left \lfloor \dfrac{rj}{m} \right \rfloor $ for $ i=1,2 $.
	It follows from Lemma~\ref{lem:Weier_gap at one place} that the sets
	\begin{align*}
	\Big\{  mk_1+j \in \mathbb{N}~\Big|
	~1\leqslant j \leqslant m-1- \left \lfloor \dfrac{m}{r} \right \rfloor,  ~0\leqslant k_1 \leqslant r-2- \left \lfloor \dfrac{rj}{m} \right \rfloor \Big	\}
	\end{align*}
	and
	\begin{align*}
	\Big\{  m k_2 +j\in \mathbb{N} ~\Big|
	~1\leqslant j \leqslant m-1- \left \lfloor \dfrac{m}{r} \right \rfloor,  ~0\leqslant k_2 \leqslant r-2- \left \lfloor \dfrac{rj}{m} \right \rfloor	\Big\}
	\end{align*} 	
	are the Weierstrass gap sets $ G(P_1) $ and $ G(P_2) $,  respectively.
	Thus
	$ \Gamma'\subseteq G(P_1)\times G(P_2) $ and the cardinality of $ \Gamma' $ is $ g $. It can be computed from Proposition~\ref{prop:divisor} that the divisor of the function
	$
	h=z^{-mk_1-j}(x-\alpha_2)^{k_1-k_2}\prod_{\mu=3}^r (x-\alpha_{\mu})^{k_1+1}
    $
	is
	\begin{align*}
	(h)=& -(mk_1+j)P_1 - (m k_2 +j)P_2 + \sum_{\mu=3}^r (m-j)P_{\mu } + \\
	&(rj+(2-r)m+m(k_1+k_2))P_{\infty}.
	\end{align*}
	When $ 1\leqslant j  \leqslant m-1- \left \lfloor \dfrac{m}{r} \right \rfloor$, the valuations of $ h $ at $ P_{\mu }~ (\mu \geqslant 3) $ are positive as $ m-j >0 $,
	and the valuation of $ h $ at $ P_{\infty} $ is also positive, because
	\begin{align*}
	& rj+(2-r)m+m(k_1+k_2)\\
	& =rj+(2-r)m+m\Bk{r-2-\left \lfloor \dfrac{rj}{m} \right \rfloor} \\
	& =m\dfrac{rj}{m} - m \left \lfloor \dfrac{rj}{m} \right \rfloor >0.
	\end{align*}
	So the pole divisor of $ h $ is $ (h)_{\infty} = (mk_1+j)P_1 + (m k_2 +j)P_2 $. Thus $ \Gamma'\subseteq H(P_1,P_2) $ and we conclude from Lemma~\ref{lem:gammaforinf} that
	$ \Gamma(P_1,P_2) = \Gamma' $.
\end{proof}

Somewhat more generally, we have the following.

\begin{thm}\label{thm:Gammapinf}
	Let $F/K$ be the Kummer extension given by~\eqref{eq:Kumext} with genus $g \geqslant 1 $, $ P_{\infty} , P_1 ,\cdots,  P_l $ be totally ramified places in $ \mathbb{P}_{F} $. Then for $ 1\leqslant l \leqslant r-\ceil{\dfrac{r}{m}} $, we have
	\begin{align*}
	\Gamma(P_{\infty}, P_1,\cdots, P_l) = \Big\{&
	(mk_0-rj, mk_1+j, \cdots, mk_l+j ) \in \mathbb{N}^{l+1} ~\Big|
	~1\leqslant j \leqslant m-1- \floor {\dfrac{m}{r}} , \\
	& ~k_1,\cdots,k_l \geqslant 0,~k_0 \geqslant \ceil {\dfrac{rj}{m}},~ k_0+k_1+\cdots +k_l=r-l	\Big\},
	\end{align*}	
	and for $ r-\ceil{\dfrac{r}{m}} < l \leqslant r $,
	\begin{equation*}
	\Gamma(P_{\infty},P_1,\cdots, P_l) =\varnothing.
	\end{equation*}
\end{thm}
\begin{proof}
	We begin by setting up some notation.
	Let $ \boldsymbol{\delta}_{\mathbf{k},j} := (mk_0-rj,mk_1+j, \cdots, mk_l+j) $ and it will only be used to describe vectors where $ 1 \leqslant j \leqslant m-1- \left \lfloor \dfrac{m}{r} \right \rfloor, k_1,\cdots,k_l \geqslant 0,~k_0 \geqslant \ceil {\dfrac{rj}{m}} $.
	For $ 2 \leqslant l \leqslant r $, we define
	\begin{align*}
	S_{\infty,l} &:= \Big\{  \boldsymbol{\delta}_{\mathbf{k},j} \in \mathbb{N}^{l+1}~\Big|~
	\sum_{i=0}^{l}k_i=r-l \Big\},\\
	\Gamma_{\infty,l} &:= \Gamma(P_{\infty},P_1,\cdots, P_l).
	\end{align*}
	 Actually, we obtain $ \ceil {\dfrac{r}{m}} \leqslant k_0 \leqslant r-l $, which gives that $ 1\leqslant l \leqslant r-\ceil{\dfrac{r}{m}} $. In the following, we will prove that $ \Gamma_{\infty,l} =S_{\infty,l}$ by induction on $ l $ for $ 1\leqslant l \leqslant r-\ceil{\dfrac{r}{m}} $. By Lemma~\ref{lem:Gammap2inf}, we have
	\begin{align*}
	\Gamma_{\infty,1} = \Big\{
	\boldsymbol{\delta}_{(k_0,k_1),j}\in\mathbb{N}^2 ~\Big|
	~ k_0+k_1=r-1 \Big\}=S_{\infty,1},
	\end{align*}
	which settles the case where $ l=1 $. We now proceed by induction on $ l \geqslant 2 $. Assume that $ \Gamma_{\infty,i} =S_{\infty,i}$ holds for all $ 1 \leqslant i \leqslant l-1 $.
	First, we claim that $ S_{\infty,l} \subseteq \Gamma_{\infty,l} $. Let
	$ \boldsymbol{\delta}_{\mathbf{k},j} \in  S_{\infty,l} $. Then
	\begin{align*}
	\left( \dfrac{\prod_{i=l+1}^r(x-\alpha_i)^{k_1+1}}{z^{mk_1+j} \prod_{i=2}^{l} (x-\alpha_i)^{k_i-k_1 } }\right)_{\infty}
	= (mk_0-rj)P_{\infty}+\sum_{i=1}^{l}(mk_i+j)P_i,
	\end{align*}
 since $ \sum_{i=0}^{l}k_i=r-l $. Hence, $ \boldsymbol{\delta}_{\mathbf{k},j} \in H_{\infty,l} $, where we write $ H(P_{\infty},P_1,\cdots,P_l) $ as $ H_{\infty,l}  $ for short.

In order to show that $ \boldsymbol{\delta}_{\mathbf{k},j} \in \Gamma_{\infty,l}   $, it suffices to prove that $ \boldsymbol{\delta}_{\mathbf{k},j} $ is minimal in $ \{ \mathbf{p}\in H_{\infty,l}~|~ p_0=mk_0-rj \} $. Suppose that $ \boldsymbol{\delta}_{\mathbf{k},j} $ is not minimal in
\begin{align*}
\{ \mathbf{p}\in H_{\infty,l}~|~ p_0=mk_0-rj \}.
\end{align*}
Then there exists $ \mathbf{u}:=(u_0,u_1,\cdots,u_l) \in H_{\infty,l} $ with $ u_0=mk_0-rj $, $ \mathbf{u} \preceq \boldsymbol{\delta}_{\mathbf{k},j} $ and $ \mathbf{u} \neq \boldsymbol{\delta}_{\mathbf{k},j} $. Let $ h \in \mathbb{F}_q(x) $ be such that $ (h)_{\infty} = u_0P_{\infty}+u_1P_1+\cdots+u_lP_l $. Note that $ \mathbf{u} \neq \boldsymbol{\delta}_{\mathbf{k},j} $ gives $ u_i < mk_i+j $ for some $ 1\leqslant i \leqslant l $. Without loss of generality, we may assume that $ u_l<mk_l+j $. Hence,
\begin{align*}
u_l= mk_l+j-t
\end{align*}
for some $ t \geqslant 1 $. In other words, we denote
\begin{align*}
\mathbf{u} = (mk_0-rj, u_1, \cdots,u_{l-1}, mk_l+j-t),
\end{align*}
where $ 0 \leqslant u_i \leqslant mk_i+j $ for all $ 1 \leqslant i \leqslant l-1 $. Thus, there are two special cases to consider:
\begin{enumerate}
	\item[(1)] $ j > t $
	\item[(2)] $ j \leqslant t $.
\end{enumerate}

Case (1): Suppose $ j > t $. In this case, we take $ v_0 := m(k_0+k_l)-rt $ and $v_i:= \max \{ u_i-(j-t),0 \} $ for $ 1 \leqslant i \leqslant l-1 $. Then
\begin{align*}
\left(  hz^{j-t} (x-\alpha_l)^{k_l}
\right)_{\infty} =
v_0 P_{\infty}  + \sum_{i=1}^{l-1} v_i P_i.
\end{align*}
Hence,
\begin{equation*}
\mathbf{v} :=
(v_0, v_1,\cdots,v_{l-1})\in H_{\infty,l-1}.
\end{equation*}

Let us introduce
\begin{equation*}
\mathbf{w} :=
(v_0, m (k_1+1)+t,mk_2+t,\cdots,mk_{l-1}+t) \in S_{\infty,l-1}.
\end{equation*}
By the induction hypothesis, $ S_{\infty,l-1}=\Gamma_{\infty,l-1}  $, and so
\begin{equation*}
\mathbf{w} \in \Gamma_{\infty,l-1}.
\end{equation*}
It follows from Lemma~\ref{lem:gamma_l} that $ \mathbf{w} $	is minimal in the set
$ \{ \mathbf{p}\in H_{\infty,l-1}~|~p_0=v_0 \} $.
Then we will find a contradiction. It is easy to check that $ v_i \leqslant mk_i+t  $  for $ 2 \leqslant i \leqslant l-1 $ and $v_1 < w_1 $ as
\begin{equation*}
u_1-(j-t) \leqslant mk_1+t <m (k_1+1)+t= w_1.
\end{equation*}
This means that $ \mathbf{v}\preceq \mathbf{w} $ and $ \mathbf{v} \neq \mathbf{w} $.

Now we have
\begin{align*}
& \mathbf{v} \in \{\mathbf{p}\in H_{\infty,l-1}~|~p_0=v_0 \},\\
& \mathbf{v} \preceq \mathbf{w}, \text{ and}  \\
& \mathbf{v} \neq \mathbf{w},
\end{align*}
which is a contradiction to the minimality of $ \mathbf{w} $  in
$ \{ \mathbf{p}\in H_{\infty,l-1}~|~p_0=v_0 \} $.

Case (2): Suppose $ j \leqslant t $. In this case, we set $ v_0 := m(k_0+k_l)-rj$. Then
\begin{equation*}
\left( h (x-\alpha_l)^{k_l}\right)_{\infty}
=v_0 P_{\infty} +\sum_{i=1}^{l-1} u_i P_i,
\end{equation*}		
which implies that
\begin{equation*}
\mathbf{v} :=
(v_0, u_1,\cdots,u_{l-1})\in H_{\infty, l-1}.
\end{equation*}
Then
\begin{equation*}
\mathbf{w} :=
(v_0, m (k_1+1)+j,mk_2+j,\cdots,mk_{l-1}+j) \in S_{\infty,l-1}.
\end{equation*}
It is easy to see that $ \mathbf{v}\preceq \mathbf{w} $ and $ \mathbf{v}\neq \mathbf{w} $, as $u_1<m(k_1+1)+j$ and $ u_i \leqslant mk_i+j  $  for $ 2 \leqslant i \leqslant l-1 $. Note that $ \mathbf{w} \in S_{\infty,l-1} $. By the induction hypothesis,
$ S_{\infty,l-1}=\Gamma_{\infty,l-1} $, and so
\begin{equation*}
\mathbf{w}  \in \Gamma_{\infty,l-1}.
\end{equation*}
By Lemma~\ref{lem:gamma_l}, $ \mathbf{w} $	is minimal in the set
$ \{ \mathbf{p}\in H_{\infty,l-1}~|~p_0=v_0 \} $.
This leads to a contradiction as
\begin{align*}
& \mathbf{v} \in \{\mathbf{p}\in H_{\infty,l-1}~|~p_0=v_0 \},\\
& \mathbf{v} \preceq \mathbf{w}, \text{ and} \\
& \mathbf{v} \neq \mathbf{w},
\end{align*}

Since both cases (1) and (2) yield a contradiction, it must be the case that
$ \boldsymbol{\delta}_{\mathbf{k},j} $ is minimal in $ \{ \mathbf{p}\in H_{\infty,l}~|~ p_0=mk_0-rj \} $. Therefore, by the definition of $ \Gamma_{\infty,l} $, we have
$ \boldsymbol{\delta}_{\mathbf{k},j} \in \Gamma_{\infty,l} $. This completes the proof of the claim that $ S_{\infty,l} \subseteq \Gamma_{\infty,l} $.

Next, we will show that  $ \Gamma_{\infty,l} \subseteq S_{\infty,l} $. Suppose not, and we suppose that there exists $ \mathbf{n}=(n_0,n_1,\cdots,n_l) \in \Gamma_{\infty,l} \setminus S_{\infty,l} $. Then there exists $ h \in \mathbb{F}_q(x) $ such that $ (h)_{\infty} = n_0P_{\infty}+ n_1P_1+\cdots+n_lP_l $. From Lemma~\ref{lem:gamma_linGapsets}, we have
\begin{equation*}
\mathbf{n} \in \Gamma_{\infty,l} \subseteq G(P_{\infty}) \times G(P_1) \times \cdots \times G(P_l).
\end{equation*}
It follows that
\begin{equation*}
\mathbf{n} =(mk_0-rj_0, mk_1+j_1,mk_2+j_2,\cdots,mk_l+j_l),
\end{equation*}
where $ 1 \leqslant j_i \leqslant m-1 $, $ k_0 \geqslant \ceil {\dfrac{rj_0}{m}} $ and $ k_i \geqslant 0 $
for all $ 0 \leqslant i \leqslant l $. Without loss of generality, we may
assume that $ j_l=\max\{ j_i~|~0 \leqslant i \leqslant l-1  \} $. Then
\begin{equation*}
\left( \dfrac{h (x-\alpha_l)^{k_l+1}}{(x-\alpha_1)^{k_l+1}}\right)_{\infty}
= n_0 P_{\infty} +(n_1+m(k_l+1))P_1 + \sum_{i=2}^{l-1}n_i P_i,
\end{equation*}	
which implies that $ (n_0,n_1+m(k_l+1), n_2, \cdots, n_{l-1})
 \in H_{\infty,l-1}\cap \mathbb{N}^l $.
By Lemma~\ref{lem:Weisemigp}, there exists $ \mathbf{\tilde{u}}=(u_0,u_1,\cdots,u_{l-1})\in \tilde{\Gamma}_{\infty,l-1} $ such that
\begin{equation}\label{eq:u1(l-1)}
\mathbf{\tilde{u}}\preceq (n_0,n_1+m(k_l+1), n_2, \cdots, n_{l-1}),
\end{equation}
and $ u_0=n_0=mk_0-rj_0>0 $. If $ u_1 \leqslant n_1 $, then $ (u_0,u_1,\cdots,u_{l-1},0) \preceq \mathbf{n} $. This yields a contradiction as
$ \mathbf{n}  $ is minimal in  $ \{ \mathbf{p}\in H_{\infty,l}~|~p_0=n_0 \} $.
Thus, we have $ u_1 > n_1 >0 $.

To get the desired conclusion, we will make use of the nonzero coordinates of $ \mathbf{\tilde{u}} $. Suppose that $ \{i_0,i_1,\cdots, i_{l-1}\}=\{0,1,\cdots,l-1\} $ is an index set such that $ 0=i_0<i_1 <i_2 <\cdots<i_k $ for some $ k $, $ 1 \leqslant k \leqslant l-1 $.  From the definition of $ \tilde{\Gamma}_{\infty,l-1}  $, we rearrange the coordinates of $ \mathbf{\tilde{u}} $, which are denoted by
\begin{align*}
u_{i_s}&=\left\{\begin{array}{lll}
m T_0-rj' && \text{ if } s=0, \\
mT_{i_s}+j' && \text{ if } 1\leqslant s \leqslant k,\\
0 && \text{ if }  k< s \leqslant l-1, \end{array}
\right.
\end{align*}
where
	$ (T_0, T_{i_1},  \cdots, T_{i_k}) \in \mathbb{N}_0^{k+1}	 $ and $ j'\in \mathbb{N} $ satisfying $ 1 \leqslant j' \leqslant m-1 $, $ T_0 \geqslant \ceil{\dfrac{rj'}{m}} $, $ T_{i_s}\geqslant 0 $ for $ 1\leqslant s \leqslant k $, and $ \sum_{s=0}^k T_{i_s} =r-k $.

 The nonzero coordinates of $ \mathbf{\tilde{u}} $ will form a new vector $ \mathbf{u} $ defined by
\begin{equation*}
{\mathbf{u}} := \boldsymbol{\delta}_{(T_{0}, T_{i_1},  \cdots, T_{i_k}),j'}
\in S_{\infty,k}.
\end{equation*}
 By the induction hypothesis, we have $ {\mathbf{u}}\in \Gamma_{\infty,k} $.

Note that $ i_1=1 $
as $ u_1 > n_1 >0 $. Since $ \gcd(r,m)=1 $ and
\begin{equation*}
mT_0-rj'=u_0=mk_0-rj_0,
\end{equation*}
 we obtain $ m |( j'-j_0) $. This forces $ j'=j_0 $ (and so $ T_0=k_0 $). As a result, we have
\begin{align*}
{\mathbf{u}} &= \boldsymbol{\delta}_{(T_0,T_1, T_{i_2} \cdots, T_{i_k}),j_0},\\
u_{i_s}& =\left\{\begin{array}{lll}
m T_0-rj_0 && \text{ if } s=0, \\
mT_{i_s}+j_0 && \text{ if } 1\leqslant s \leqslant k,\\
0 && \text{ if }  k< s \leqslant l-1, \end{array}
\right.
\end{align*}
with $ T_0+T_1+ T_{i_2}+ \cdots+ T_{i_k}= r-k $.
In particular, $ u_1=mT_1+j_0 $. In the following, we separate the proof into two cases:
\begin{enumerate}
	\item[(1)] $ u_1-m(k_l+1) \geqslant 0 $,
	\item[(2)] $ u_1-m(k_l+1) <0 $.
\end{enumerate}

Case (1): Suppose $ u_1-m(k_l+1) \geqslant 0 $. If $ u_1-m(k_l+1) = 0 $, then $ m | j_0 $, which is a contradiction. Therefore,
it must be $ u_1-m(k_l+1) >0 $. Set
\begin{equation*}
\tilde{\mathbf{v}} :=(u_0, u_1-m(k_l+1),u_2,\cdots,u_{l-1},mk_l+j_0).
\end{equation*}
By Equation~\eqref{eq:u1(l-1)}, and since $ j_0 \leqslant j_l $, we have
$ \tilde{\mathbf{v}}\preceq \mathbf{n} $. Denote
\begin{equation*}
\mathbf{v} :=\boldsymbol{\delta}_{(T_0,T_1-k_l-1, T_{i_2},   \cdots, T_{i_k},k_l),j_0}.
\end{equation*}	
As before, $ \mathbf{v} $ is formed by some of the nonzero coordinates of $ \tilde{\mathbf{v}} $. We verify that $ T_0+(T_1-k_l-1)+ \sum_{s=2}^{k} T_{i_s}+k_l =r-(k+1) $, which implies that $ \mathbf{v} \in S_{\infty,k+1} $. Since $S_{\infty,k+1} \subseteq \Gamma_{\infty,k+1} $, it follows that
$ \tilde{\mathbf{v}} \in \tilde{\Gamma}_{\infty,l} \subseteq H_{\infty,l} $. Notice that
$ \tilde{\mathbf{v}}\preceq \mathbf{n} $ and $ \mathbf{n} \in \Gamma_{\infty,l}  $. Therefore, $ \mathbf{n} = \tilde{\mathbf{v}} $ as otherwise $ \mathbf{n} $ is not minimal in $ \{ \mathbf{p}\in H_{\infty,l}~|~p_0=n_0 \} $. As none of $ n_i $ equals zero, we get that
$ \mathbf{n} = \tilde{\mathbf{v}}= \mathbf{v}\in S_{\infty,l}  $, which is a contradiction.

Case (2): Suppose that $ u_1-m(k_l+1) <0 $. At this point, we consider separately two subcases:
\begin{enumerate}
	\item[(a)] $ k_1 > 0 $,
	\item[(b)] $ k_1 = 0 $.
\end{enumerate}
Subcase (a): Suppose $ k_1 > 0 $. Let
\begin{equation*}
\tilde{\mathbf{v}} :=(u_0, m(k_1-1)+j_0,u_2,\cdots,u_{l-1},m(T_1-k_1)+j_0).
\end{equation*}
We are going to show that $ \tilde{\mathbf{v}}\preceq \mathbf{n} $ and $ \tilde{\mathbf{v}}\neq \mathbf{n} $.    	
At first sight, Equation~\eqref{eq:u1(l-1)} means that $ u_i \leqslant n_i  $ for $ 2\leqslant i \leqslant l-1 $. Since
$ j_0-m<0< j_1 $, we have
\begin{equation*}
v_1=m(k_1-1)+j_0 < mk_1+j_1=n_1.
\end{equation*}
Note that $ u_1=mT_1+j_0 < m(k_l+1) $.
This implies that
\begin{equation*}
v_l=m(T_1-k_1)+j_0 \leqslant mT_1+j_0-m < mk_l < mk_l+j_l=n_l.
\end{equation*}
So  $ \tilde{\mathbf{v}}\preceq \mathbf{n} $ and $ \tilde{\mathbf{v}}\neq \mathbf{n} $.	 Let us express $ \mathbf{v} $ as
\begin{equation*}
\mathbf{v} :=\boldsymbol{\delta}_{(T_0,k_1-1, T_{i_2},  \cdots, T_{i_k}, T_1-k_1),j_0}.
\end{equation*}
We claim that $ \mathbf{v} \in S_{\infty,k+1} $.
Clearly, $ T_0=k_0 $, $ k_1-1 \geqslant 0 $. Suppose that $ T_1-k_1 <0 $, then it must be
\begin{equation*}
u_1=mT_1+j_0 \leqslant mk_1+j_1 =n_1,
\end{equation*} contradicting the fact that $ u_1>n_1 $.
Therefore, $ T_1-k_1 \geqslant 0 $. Moreover, it is easy to check that
\begin{equation*}
T_0+(k_1-1)+  \sum_{s=2}^{k}T_{i_s}+(T_1-k_1)=r-(k+1).
\end{equation*}
So $ \mathbf{v} \in S_{\infty,k+1} \subseteq \Gamma_{\infty,k+1} $.
It follows that $ \tilde{\mathbf{v}} \in H_{\infty,l} $. Hence we have
\begin{align*}
& \tilde{\mathbf{v}} \in \{\mathbf{p}\in H_{\infty,l}~|~p_0=n_0 \},\\
& \tilde{\mathbf{v}} \preceq \mathbf{n}, \text{ and} \\
& \tilde{\mathbf{v}} \neq \mathbf{n},
\end{align*}
contradicting the minimality of $ \mathbf{n} $  in $\{ \mathbf{p}\in H_{\infty,l}~|~p_0=n_0 \} $. The proof in this subcase is completed.

Subcase (b): Suppose $ k_1 = 0 $. Set
\begin{equation*}
\tilde{\mathbf{v}} :=(u_0,0,u_2,\cdots,u_{l-1},mT_1+j_0).
\end{equation*} 	
Since $ u_1=mT_1+j_0 < m(k_l+1) $ means that $ T_1 \leqslant k_l  $, we have
\begin{equation*}
v_l=mT_1+j_0 \leqslant mk_l+j_l=n_l,
\end{equation*}	
as $ j_0 \leqslant j_l $. This yields that $ \tilde{\mathbf{v}}\preceq \mathbf{n} $ and $ \tilde{\mathbf{v}}\neq \mathbf{n} $. Let
\begin{equation*}
\mathbf{v} :=\boldsymbol{\delta}_{(T_0, T_1, T_{i_2},  \cdots, T_{i_k} ),j_0}.
\end{equation*}	
It is easy to see that $ \mathbf{v} \in S_{\infty,k}$ as
$ \sum_{s=0}^{k} T_{i_s} =r-k $. As before, it follows that $ \tilde{\mathbf{v}} \in H_{\infty,l} $ and $ \tilde{\mathbf{v}} \in \{ \mathbf{p}\in H_{\infty,l}~|~p_0=n_0 \} $. But $ \tilde{\mathbf{v}}\preceq \mathbf{n} $ and $ \tilde{\mathbf{v}}\neq \mathbf{n} $. This contradicts the minimality of $ \mathbf{n} $  in $\{ \mathbf{p}\in H_{\infty,l}~|~p_0=n_0 \} $, which finishes the proof in this subcase.

Since both cases (1) and (2) yield a contradiction, it must be te case that no such $\mathbf{n} $ exists. Hence, $ \Gamma_{\infty,l} \setminus S_{\infty,l}  =\varnothing $. This establishes that $ \Gamma_{\infty,l}  \subseteq S_{\infty,l}  $, concluding the proof that
$ \Gamma_{\infty,l}  = S_{\infty,l}  $.

Now suppose that $ r-\ceil{\dfrac{r}{m}} < l \leqslant r $. If $ \boldsymbol{\delta}_{\mathbf{k},j} \in \Gamma_{\infty,l} $, then
\begin{equation*}
\ceil{\dfrac{rj}{m}}\leqslant \sum_{i=1}^{l} k_i = r-l < \ceil{\dfrac{r}{m}},
\end{equation*}	
which is a contradiction as $ 1 \leqslant j \leqslant m-1 $. Therefore, $ \Gamma_{\infty,l} =\varnothing $.
\end{proof}

The above results enable one to deal with the general case for arbitrary distinct places excluding the place at infinity.

\begin{thm}\label{thm:lpoint}
	Let $F/K$ be the Kummer extension given by~\eqref{eq:Kumext} with genus $g \geqslant 1 $, $ P_1 ,\cdots,  P_l $ be totally ramified places in $ \mathbb{P}_{F} $. Then for $ 2\leqslant l \leqslant r-\floor{\dfrac{r}{m}} $,
	\begin{align*}
	\Gamma(P_1,P_2,\cdots, P_l) = \Big\{&
	(mk_1+j, mk_2+j,\cdots, mk_l+j) \in \mathbb{N}^l ~\Big|
	~1\leqslant j \leqslant m-1- \floor {\dfrac{m}{r}} , \\
	& ~k_1,\cdots,k_l \geqslant0,~ k_1+\cdots +k_l=r-l- \floor {\dfrac{rj}{m}} 	\Big\},
	\end{align*}
	and for $ r-\floor{\dfrac{r}{m}} < l \leqslant r $,
	\begin{equation*}
	\Gamma(P_1,P_2,\cdots, P_l) =\varnothing.
	\end{equation*}
\end{thm}

\begin{proof}
	The proof of this theorem is rather technical though it is similar to that of
	Theorem~\ref{thm:Gammapinf}, and so is showed in details here.
	Let $ \boldsymbol{\gamma}_{\mathbf{k},j} := (mk_1+j, mk_2+j,\cdots, mk_l+j) $ and we emphasize that it will only be used to describe vectors where $ 1 \leqslant j \leqslant m-1- \left \lfloor \dfrac{m}{r} \right \rfloor, ~k_1,\cdots,k_l \geqslant 0 $.
	For $ 2 \leqslant l \leqslant r $, set
	\begin{align*}
	S_l &:= \Big\{  \boldsymbol{\gamma}_{\mathbf{k},j} \in \mathbb{N}^{l}~\Big|~
	 \sum_{i=1}^{l}k_i=r-l-  \floor{ \dfrac{rj}{m} } \Big\},\\
	 \Gamma_l &:=\Gamma(P_1,P_2,\cdots, P_l).
	  \end{align*}
	 Then one have
	 $ 2\leqslant l \leqslant r-\floor{\dfrac{r}{m}} $, as $ 0 \leqslant  \sum_{i=1}^{l}k_i=r-l- \floor{\dfrac{rj}{m}} $. In the following, we will prove that $ \Gamma_l =S_l$ by induction on $ l $ for $ 2 \leqslant l \leqslant r-\floor{\dfrac{r}{m}} $. By Theorem~\ref{thm:Gamma2points}, one can see that
		\begin{align*}
		\Gamma_2 = \Big\{
		\boldsymbol{\gamma}_{(k_1,k_2),j}\in\mathbb{N}^2 ~\Big|
           ~ k_1+k_2=r-2- \left \lfloor \dfrac{rj}{m} \right \rfloor \Big\}=S_2,
		\end{align*}
	    which settles the case where $ l=2 $. We now proceed by induction on $ l \geqslant 3 $. Assume that $ \Gamma_i =S_i$ holds for all $ 2 \leqslant i \leqslant l-1 $.
	    First, we claim that $ S_l \subseteq \Gamma_l $. Let
	    $ \boldsymbol{\gamma}_{\mathbf{k},j} \in  S_l $. Then
	    \begin{align*}
	    \left( \dfrac{\prod_{i=l+1}^r(x-\alpha_i)^{k_1+1}}{z^{mk_1+j} \prod_{i=2}^{l} (x-\alpha_i)^{k_i-k_1 } }\right)_{\infty}
	    = \sum_{i=1}^{l}(mk_i+j)P_i,
	    \end{align*}
	    since $ \sum_{i=1}^{l}k_i=r-l-  \floor{ \dfrac{rj}{m} } $.
     Hence, $ \boldsymbol{\gamma}_{\mathbf{k},j} \in H_l $,  where $ H(P_1,\cdots,P_l) $ is denoted by $ H_l  $ for convenience.
	
	    In order to show that $ \boldsymbol{\gamma}_{\mathbf{k},j} \in \Gamma_l   $, it suffices to prove that $ \boldsymbol{\gamma}_{\mathbf{k},j} $ is minimal in $ \{ \mathbf{p}\in H_l~|~ p_1=mk_1+j \} $. Suppose that $ \boldsymbol{\gamma}_{\mathbf{k},j} $ is not minimal in
	    \begin{align*}
	    \{ \mathbf{p}\in H_l~|~ p_1=mk_1+j \}.
	    \end{align*}
	    Then there exists $ \mathbf{u} =(u_1,\cdots,u_l)\in H_l $ with $ u_1=mk_1+j $, $ \mathbf{u} \preceq \boldsymbol{\gamma}_{\mathbf{k},j} $ and $ \mathbf{u} \neq \boldsymbol{\gamma}_{\mathbf{k},j} $. Let $ h \in \mathbb{F}_q(x) $ be such that $ (h)_{\infty} = u_1P_1+\cdots+u_lP_l $. Also notice that $ v_{P_{\infty}}(h) \geqslant 0 $. Without loss of generality, we may assume that $ u_l<mk_l+j $ as $ \mathbf{u} \neq \boldsymbol{\gamma}_{\mathbf{k},j} $ gives $ u_i < mk_i+j $ for some $ 2\leqslant i \leqslant l $. Hence,
	    \begin{align*}
	    u_l= mk_l+j-t
	    \end{align*}
	    for some $ t \geqslant 1 $. In other words, we denote
	    \begin{align*}
	    \mathbf{u} = (mk_1+j, u_2,\cdots,u_{l-1}, mk_l+j-t),
	    \end{align*}
	    where $ 0 \leqslant u_i \leqslant mk_i+j $ for all $ 2\leqslant i \leqslant l-1 $. Thus, there are two cases to consider:
	    \begin{enumerate}
	    	\item[(1)] $ j > t $
	    	\item[(2)] $ j \leqslant t $.
	    \end{enumerate}

	 Case (1): Suppose $ j > t $. In this case, we take $ v_0 :=\max\{  r(j-t)+mk_l-v_{P_{\infty}}(h) ,0\} $, $v_i= \max \{ u_i-(j-t),0 \} $ for $ 2 \leqslant i \leqslant l-1 $. Then
	 \begin{align*}
	 \left(  hz^{j-t} (x-\alpha_l)^{k_l}
	 \right)_{\infty} =
	 v_0 P_{\infty}  +(mk_1+t)P_1+ \sum_{i=2}^{l-1} v_i P_i.
	 \end{align*}
	 There are two special subcases to consider:
	  \begin{enumerate}
	       	\item[(a)] $  v_0 > 0 $,
	       	\item[(b)] $  v_0 = 0 $.
	  \end{enumerate}
	
	 Subcase (a): Suppose $  v_0 > 0 $. Then
	 \begin{equation*}
	 \mathbf{v} :=
	 (v_0, mk_1+t,v_2,\cdots,v_{l-1})\in H_{\infty,l-1}.
	 \end{equation*}	
	 Set \begin{equation*}
	 \rho_0 :=k_l+1+ \floor{\dfrac{rj}{m}} .
	 \end{equation*}
	 Since $ \sum_{i=1}^{l}k_i=r-l-  \floor{ \dfrac{rj}{m} } $, it follows from Theorem~\ref{thm:Gammapinf} that
	 \begin{equation*}
	 \mathbf{w} :=
	 (m\rho_0-rt, m k_1+t,\cdots,mk_{l-1}+t) \in S_{\infty,l-1}.
	 \end{equation*}
	 Note that $ S_{\infty,l-1}=\Gamma_{\infty,l-1}  $. Hence
	 \begin{equation*}
	 \mathbf{w} \in \Gamma_{\infty,l-1}.
	 \end{equation*}
	 By Lemma~\ref{lem:gamma_l}, $ \mathbf{w} $	is minimal in the set
	 $ \{ \mathbf{p}\in H_{\infty,l-1}~|~p_1=mk_1+t \} $.
	
	 In the following, we are going to show that $ \mathbf{v}\preceq \mathbf{w} $. One easily check that $ v_i \leqslant mk_i+t =w_i $  for $ 2 \leqslant i \leqslant l-1 $.
	Since
	 \begin{equation*}
      0 \leqslant \dfrac{rj}{m} -\floor{\dfrac{rj}{m}} <1,
	 \end{equation*}
	 we obtain $ v_0 \leqslant r(j-t)+mk_l < m\rho_0-rt=w_0 $.

	 Now we have
	 \begin{align*}
	 & \mathbf{v} \in \{\mathbf{p}\in H_{\infty, l-1}~|~p_1=mk_1+t \},\\
	 & \mathbf{v} \preceq \mathbf{w}, \text{ and} \\
	 & \mathbf{v} \neq \mathbf{w},
	 \end{align*}
	 which is a contradiction to the minimality of $ \mathbf{w} $  in
	 $ \{ \mathbf{p}\in H_{\infty,l-1}~|~p_1=mk_1+t \} $.
	
	Subcase (b): Suppose $  v_0 = 0 $. It follows that
	\begin{align*}
	\mathbf{v} &:=
	(mk_1+t,v_2,v_3,\cdots,v_{l-1})\in H_{l-1},\\
	\mathbf{w} &:=
	( m k_1+t,m \rho_2+t,mk_3+t, \cdots,mk_{l-1}+t) \in S_{l-1},
	\end{align*}
	where $ \rho_2 :=k_2+k_l+1 $.
	By the induction hypothesis, $ S_{l-1}=\Gamma_{l-1}  $, and so
	\begin{equation*}
	\mathbf{w} \in \Gamma_{l-1}.
	\end{equation*}
	From Lemma~\ref{lem:gamma_l}, $ \mathbf{w} $	is minimal in the set
	$ \{ \mathbf{p}\in H_{l-1}~|~p_1=mk_1+t \} $.

	Similarly, we have
	\begin{align*}
	& \mathbf{v} \in \{\mathbf{p}\in H_{ l-1}~|~p_1=mk_1+t \},\\
	& \mathbf{v} \preceq \mathbf{w}, \text{ and} \\
	& \mathbf{v} \neq \mathbf{w},
	\end{align*}
	contradicting the minimality of $ \mathbf{w} $  in
	$ \{ \mathbf{p}\in H_{l-1}~|~p_1=mk_1+t \} $.
	
	Case (2): Suppose $ j \leqslant t $. In this case, we have
	\begin{equation*}
	\left( \dfrac{h (x-\alpha_l)^{k_l}}{(x-\alpha_2)^{k_l}}\right)_{\infty}
	=(mk_1+j)P_1+(u_2+mk_l)P_2 +\sum_{i=3}^{l-1} u_i P_i,
	\end{equation*}		
	which implies that
		 \begin{equation*}
		 \mathbf{v} :=
		 (mk_1+j, u_2+mk_l,u_3,\cdots,u_{l-1})\in H_{l-1}.
		 \end{equation*}
	Then
		 \begin{equation*}
		 \mathbf{w} :=
		 (mk_1+j, m \rho_2+j,mk_3+j,\cdots,mk_{l-1}+j) \in S_{l-1},
		 \end{equation*}
		 where 	$ \rho_2 =k_2+k_l+1 $.
	We will find a contradiction if we show that $ \mathbf{v}\preceq \mathbf{w} $ and $ \mathbf{v}\neq \mathbf{w} $. Clearly, $ v_i=u_i \leqslant mk_i+j =w_i $  for $ 3 \leqslant i \leqslant l-1 $.
	It suffices to prove that $ u_2+mk_l < m \rho_2+j $. But this inequality always holds as
	\begin{equation*}
	 u_2+mk_l \leqslant m k_2+j+mk_l <  m \rho_2+j .
	\end{equation*}	
	
	Since both cases (1) and (2) yield a contradiction, it must be the case that
	$ \boldsymbol{\gamma}_{\mathbf{k},j} $ is minimal in $ \{ \mathbf{p}\in H_l~|~ p_1=mk_1+j \} $. Therefore, by the definition of $ \Gamma_l $, we have
	$ \boldsymbol{\gamma}_{\mathbf{k},j} \in \Gamma_l $. This completes the proof of the claim that $ S_l \subseteq \Gamma_l $.
	
	Next, we will show that  $ \Gamma_l \subseteq S_l $. Suppose not, and we suppose that there exists $ \mathbf{n}=(n_1,\cdots,n_l) \in \Gamma_l \setminus S_l $. Then there exists $ h \in \mathbb{F}_q(x) $ such that $ (h)_{\infty} = n_1P_1+\cdots+n_lP_l $. By Lemma~\ref{lem:gamma_linGapsets}, we have
	\begin{equation*}
	\mathbf{n} \in \Gamma_l \subseteq G(P_1) \times \cdots \times G(P_l).
	\end{equation*}
	It follows that
		\begin{equation*}
		\mathbf{n} =(mk_1+j_1,mk_2+j_2,\cdots,mk_l+j_l),
		\end{equation*}
	where $ 1 \leqslant j_i \leqslant m-1 $ and $ k_i \geqslant 0 $
	for all $ 1 \leqslant i \leqslant l $. Without loss of generality, we may
	assume that $ j_l=\max\{ j_i~|~2 \leqslant i \leqslant l-1  \} $. Then
		\begin{equation*}
		\left( \dfrac{h (x-\alpha_l)^{k_l+1}}{(x-\alpha_{1})^{k_l+1}}\right)_{\infty}
		= (n_{1}+m(k_l+1))P_{1} + \sum_{i=2}^{l-1}n_i P_i,
		\end{equation*}	
	   which implies that $ (n_{1}+m(k_l+1), n_2, \cdots, n_{l-1}) \in H_{l-1} \cap \mathbb{N}^{l-1}$.
	   By Lemma~\ref{lem:Weisemigp}, there exists $ \tilde{\mathbf{u}}=(u_1,\cdots,u_{l-1})\in \tilde{\Gamma}_{l-1} $ such that
	\begin{equation}\label{eq:u(l-1)}
	\tilde{\mathbf{u}}\preceq (n_{1}+m(k_l+1), n_2, \cdots, n_{l-1}),
	\end{equation}
	and $ u_2=n_2=mk_2+j_2 $. If $ u_1 \leqslant n_1 $, then $ (u_1,\cdots,u_{l-1},0) \preceq \mathbf{n} $. This yields a contradiction as
	$ \mathbf{n}  $ is minimal in the set $ \{ \mathbf{p}\in H_{l}~|~p_2=n_2 \} $.
	Thus, we have $ u_1 > n_1 >0 $. By the induction hypothesis,
		\begin{equation*}
		\mathbf{u} := \boldsymbol{\gamma}_{(T_{i_1},  \cdots, T_{i_k}),j'}
		\in S_k=\Gamma_k,
		\end{equation*}
	for some $ k $, $ 2 \leqslant k \leqslant l-1 $. Remember that $ \mathbf{u} $ is constructed from the nonzero coordinates of $ \tilde{\mathbf{u}} $ and $ \sum_{s=1}^k T_{i_s} =r-k- \floor{\dfrac{rj'}{m}} $.

Let $ \{i_1,\cdots, i_{l-1}\}=\{1,\cdots,l-1\} $ be an index set such that $ i_1 <i_2 <\cdots<i_k $ and
	\begin{align*}
	u_{i_s}=\left\{\begin{array}{lll}
		mT_{i_s}+j' && \text{ if } 1\leqslant s \leqslant k, \\
		0     &&    \text{ if } k < s \leqslant l-1.
		\end{array}
		\right.
	\end{align*}
 Note that $ i_1=1 $ and $ i_2=2 $
 since $ u_1 > n_1 >0 $ and $ u_2=n_2>0 $. The fact
 \begin{equation*}
 mT_2+j'=u_2=mk_2+j_2
 \end{equation*} gives that $ m $ divides $ j'-j_2 $. This forces $ j'=j_2 $, $ T_2=k_2 $. So we have
 		\begin{align*}
 		{\mathbf{u}} &= \boldsymbol{\gamma}_{(T_1,T_2, T_{i_3} \cdots, T_{i_k}),j_2},\\
 		u_{i_s}& =\left\{\begin{array}{lll}
 		mT_{i_s}+j_2 && \text{ if } 1\leqslant s \leqslant k, \\
 		0     &&    \text{ if } k < s \leqslant l-1,
 		\end{array}
 		\right.
 		\end{align*}
      with $ T_1+T_2+ \sum_{s=3}^k T_{i_s}= r-k- \floor{\dfrac{rj_2}{m}} $.
      Specifically, $ u_1=mT_1+j_2 $. In the following, we seperate the proof into two cases:
      	    \begin{enumerate}
      	    	\item[(1)] $ u_1-m(k_l+1) \geqslant 0 $,
      	    	\item[(2)] $ u_1-m(k_l+1) <0 $.
      	    \end{enumerate}
 	Case (1): Suppose $ u_1-m(k_l+1) \geqslant 0 $. Since $ m \nmid j_2 $,
 	it follows that $ u_1-m(k_l+1) >0 $. Set
 			\begin{equation*}
 			\tilde{\mathbf{v}} :=(u_1-m(k_l+1),u_2,\cdots,u_{l-1},mk_l+j_2).
 			\end{equation*}
 	By Equation~\eqref{eq:u(l-1)}, and since $ j_2 \leqslant j_l $, we have
 	$ \tilde{\mathbf{v}}\preceq \mathbf{n} $. Denote
  	 \begin{equation*}
  	 \mathbf{v} :=\boldsymbol{\gamma}_{(T_1-k_l-1, T_2, T_{i_3},  \cdots, T_{i_k},k_l),j_2}.
  	 \end{equation*}	
 	Since $ (T_1-k_l-1)+ T_2+\sum_{s=3}^{k} T_{i_s}+k_l =r-(k+1)- \floor{\dfrac{rj_2}{m}}  $, this implies that $ \mathbf{v} \in S_{k+1} $. Since $S_{k+1} \subseteq \Gamma_{k+1} $, it follows that
 	 $ \tilde{\mathbf{v}} \in \tilde{\Gamma}_l \subseteq H_l $. Notice that
 	  $ \tilde{\mathbf{v}}\preceq \mathbf{n} $ and $ \mathbf{n} \in \Gamma_l   $. Therefore, $ \mathbf{n} = \tilde{\mathbf{v}} $ as otherwise $ \mathbf{n} $ is not minimal in $ \{ \mathbf{p}\in H_{l}~|~p_2=n_2 \} $. Hence, $ k+1=l $ and
 	  $ \mathbf{n} = \tilde{\mathbf{v}}= \mathbf{v}\in S_{l}  $, which is a contradiction.
 	
 	Case (2): Suppose that $ u_1-m(k_l+1) <0 $. There are two subcases to consider:
      \begin{enumerate}
         \item[(a)] $ k_1 > 0 $,
         \item[(b)] $ k_1 = 0 $.
      \end{enumerate}
     Subcase (a): Suppose $ k_1 > 0 $. Let
   \begin{equation*}
     \tilde{\mathbf{v}} :=(m(k_1-1)+j_2,u_2,\cdots,u_{l-1},m(T_1-k_1)+j_2).
   \end{equation*}
   We are going to show that $ \tilde{\mathbf{v}}\preceq \mathbf{n} $ and $ \tilde{\mathbf{v}}\neq \mathbf{n} $.    	
	Obviously, $ u_i \leqslant n_i $	for $ 2\leqslant i \leqslant l-1 $. Since
	$ j_2-m<0< j_1 $, we have
	\begin{equation*}
	   v_1=m(k_1-1)+j_2 < mk_1+j_1.
	 \end{equation*}
	The fact that $ u_1=mT_1+j_2 < m(k_l+1) $ gives
	\begin{equation*}
	v_l=m(T_1-k_1)+j_2 \leqslant mT_1+j_2-m < mk_l <mk_l+j_l= n_l.
	\end{equation*}
	So  $ \tilde{\mathbf{v}}\preceq \mathbf{n} $ and $ \tilde{\mathbf{v}}\neq \mathbf{n} $. The nonzero coordinates of $ \tilde{\mathbf{v}} $	will form a new vector expressed as
	\begin{equation*}
	\mathbf{v} :=\boldsymbol{\gamma}_{(k_1-1, T_2, T_{i_3},  \cdots, T_{i_k}, T_1-k_1),j_2}.
	\end{equation*}
	We claim that $ \mathbf{v} \in S_{k+1} $.
	Clearly, $ k_1-1 \geqslant 0 $. Suppose that $ T_1-k_1 <0 $, then it must be
	\begin{equation*}
	u_1=mT_1+j_2 \leqslant mk_1+j_1 =n_1,
	\end{equation*} contradicting the fact that $ u_1>n_1 $.
	Therefore, $ T_1-k_1 \geqslant 0 $. It is easy to see that
	\begin{equation*}
      (k_1-1)+ T_2+ \sum_{s=3}^{k}T_{i_s}+(T_1-k_1)=r-(k+1)-\floor{\dfrac{rj_2}{m}}.
	\end{equation*}
	So $ \mathbf{v} \in S_{k+1} \subseteq \Gamma_{k+1} $.
	Then $ \tilde{\mathbf{v}} \in H_l $ and so $ \tilde{\mathbf{v}} \in \{ \mathbf{p}\in H_{l}~|~p_2=n_2 \} $. This contradicts the minimality of $ \mathbf{n} $  in $\{ \mathbf{p}\in H_{l}~|~p_2=n_2 \} $, concluding the proof in this subcase.
	
	Subcase (b): Suppose $ k_1 = 0 $. Set
   \begin{equation*}
   \tilde{\mathbf{v}} :=(0,u_2,\cdots,u_{l-1},mT_1+j_2).
   \end{equation*} 	
	Since $ u_1=mT_1+j_2 < m(k_l+1) $ means that $ T_1 \leqslant k_l  $, we have
   \begin{equation*}
   v_l=mT_1+j_2 \leqslant mk_l+j_l=n_l,
   \end{equation*}	
	as $ j_2 \leqslant j_l $. This yields that $ \tilde{\mathbf{v}}\preceq \mathbf{n} $ and $ \tilde{\mathbf{v}}\neq \mathbf{n} $. Let
	\begin{equation*}
	\mathbf{v} :=\boldsymbol{\gamma}_{(T_1, T_2, T_{i_3},  \cdots, T_{i_k} ),j_2}.
	\end{equation*}	
	One knows that $ \mathbf{v} \in S_{k}$ as
	 $ \sum_{s=1}^{k} T_{i_s} =r-k- \floor{\dfrac{rj_2}{m}}  $. As before, it follows that $ \tilde{\mathbf{v}} \in H_l $ and $ \tilde{\mathbf{v}} \in \{ \mathbf{p}\in H_{l}~|~p_2=n_2 \} $. But $ \tilde{\mathbf{v}}\preceq \mathbf{n} $ and $ \tilde{\mathbf{v}}\neq \mathbf{n} $, contradicting the minimality of $ \mathbf{n} $  in $\{ \mathbf{p}\in H_{l}~|~p_2=n_2 \} $. The proof in this subcase is completed.
	
	 Since both cases (1) and (2) yield a contradiction, it must be the case that no such $\mathbf{n} $ exists. Hence, $ \Gamma_l \setminus S_l =\varnothing $. This establishes that $ \Gamma_l \subseteq S_l $, concluding the proof that
	 $ \Gamma_l = S_l $.
	
	  Now suppose that $ r-\floor{\dfrac{r}{m}} < l \leqslant r $. If $ \boldsymbol{\gamma}_{\mathbf{k},j} \in \Gamma_l $, then
	\begin{equation*}
	0 \leqslant \sum_{i=1}^{l} k_i = r-l-\floor{\dfrac{rj}{m}} < \floor{\dfrac{r}{m}}-\floor{\dfrac{rj}{m}},
	\end{equation*}	
	which is a contradiction as $ 1 \leqslant j \leqslant m-1 $. Therefore, $ \Gamma_l =\varnothing $.
\end{proof}

 \begin{remark}
 We mention that Theorem~\ref{thm:lpoint} is an extension of Theorem 10 in~\cite{matthews2004weierstrass}, which settles the case where $ r=q $, $ m=q+1 $.
 \end{remark}

\section{Examples over Kummer extensions}\label{sec:Examples}

 In this section we provide several examples to illustrate our results given in the previous section.

\begin{example}
 Let $ (r,m,\lambda)=(7,5,1) $. The Kummer extension $F_{\mathcal{K}}/\mathbb{F}_q(x) $ is defined by $ y^5=f(x) $ with $ \deg(f)=7 $, where $ q $ is a prime power.
 Let
 \begin{equation*}
 (P_1,P_2,\cdots,P_7)
 \end{equation*} be a $ 7 $-tuple of totally ramified places with the exception of $ P_{\infty} $.
By Theorem~\ref{thm:Gamma2points},
\begin{align*}
\Gamma_2=\left\{\begin{array}{ll}
(1, 21), (2, 17),(3, 8),(4, 4),
(6, 16), (7, 12), (8, 3), \\
(11, 11),(12, 7),
(16, 6), (17, 2),
(21, 1)
\end{array}
 \right\}.
\end{align*}
Applying Theorem~\ref{thm:lpoint}, one can obtain that

\begin{align*}
\Gamma_3&=\left\{\begin{array}{ll}
  (1, 1, 16),
  (1, 6, 11),
  (1, 11, 6),
  (1, 16, 1),
    (2, 2, 12),
    (2, 7, 7), \\
    (2, 12, 2),
   (3, 3, 3),
  (6, 1, 11),
  (6, 6, 6),
  (6, 11, 1),
   (7, 2, 7), \\
  (7, 7, 2),
  (11, 1, 6),
  (11, 6, 1),
    (12, 2, 2),
  (16, 1, 1)
\end{array}
\right\},\\
\Gamma_4&=\left\{\begin{array}{ll}
 (1, 1, 1, 11),
 (1, 1, 6, 6),
 (1, 1, 11, 1),
 (1, 6, 1, 6),
 (1, 6, 6, 1), \\
 (1, 11, 1, 1),
  (2, 2, 2, 7),
  (2, 2, 7, 2),
  (2, 7, 2, 2),
  (7, 2, 2, 2),  \\
 (6, 1, 1, 6),
 (6, 1, 6, 1),
 (6, 6, 1, 1),
 (11, 1, 1, 1)
\end{array}
\right\},\\
\Gamma_5&=\left\{\begin{array}{ll}
 (1, 1, 1, 1, 6),
 (1, 1, 1, 6, 1),
 (1, 1, 6, 1, 1), \\
 (1, 6, 1, 1, 1),
 (2, 2, 2, 2, 2),
 (6, 1, 1, 1, 1)
 \end{array}
\right\},\\
\Gamma_6&=\left\{
(1, 1, 1, 1,  1, 1)
\right\},\\
\Gamma_7&=\varnothing.
\end{align*}
\end{example}

\begin{example}
	Let $ (r,m,\lambda)=(5,9,1) $. The Kummer extension $F_{\mathcal{K}}/\mathbb{F}_q(x) $ is defined by $ y^9=f(x) $ with $ \deg(f)=5 $, where $ q $ is a prime power.
	Let
	\begin{equation*}
	(P_{\infty},P_1,P_2,\cdots,P_5)
	\end{equation*} be a $ 6 $-tuple of totally ramified places.
	
It follows from Lemma~\ref{lem:Gammap2inf} that
	\begin{align*}
	\Gamma_{\infty,1}=\left\{\begin{array}{ll}
(1, 7),
(2, 14),
(3, 21),
(4, 28),
(6, 6),
(7, 13),
(8, 20),\\
(11, 5),
(12, 12),
(13, 19),
(16, 4),
(17, 11),
(21, 3), \\
(22, 10),
(26, 2),
(31, 1)
	\end{array}
	\right\}.
	\end{align*}	
Applying Theorem~\ref{thm:Gammapinf}, one can obtain that

\begin{align*}
\Gamma_{\infty,2}&=\left\{\begin{array}{ll}	
   (2, 5, 5),
   (3, 3, 12),
  (3, 12, 3),
  (4, 1, 19),
 (4, 10, 10), \\
 (4, 19, 1),
 (7, 4, 4),
 (8, 2, 11),
 (8, 11, 2),
 (12, 3, 3), \\
 (13, 1, 10),
 (13, 10, 1),
 (17, 2, 2),
 (22, 1, 1)
	\end{array}
	\right\},\\
\Gamma_{\infty,3}&=\left\{\begin{array}{ll}	
 (3, 3, 3, 3),
 (4, 1, 1, 10),
 (4, 1, 10, 1), \\
 (4, 10, 1, 1),
  (8, 2, 2, 2),
  (13, 1, 1, 1)
	\end{array}
	\right\},\\	
\Gamma_{\infty,4}&=\left\{
  (4 ,    1  ,   1  ,   1  ,   1)
\right\},\\
\Gamma_{\infty,5}&=\varnothing.
\end{align*}	
\end{example}


%





\ifCLASSOPTIONcaptionsoff
  \newpage
\fi



%
\bibliographystyle{IEEEtran}
\bibliography{paper}

%






\end{document}